\documentclass[conference]{IEEEtran}
\IEEEoverridecommandlockouts

\usepackage{cite}
\usepackage{amsmath,amssymb,amsfonts}
\usepackage[ruled,vlined]{algorithm2e}
\usepackage{graphicx}
\usepackage{textcomp}
\usepackage{xcolor}
\usepackage{multirow}
\usepackage{hyperref}
\usepackage{comment}
\definecolor{linkcolor}{RGB}{83,83,182}
\hypersetup{
    colorlinks=true,
    citecolor=linkcolor,
    linkcolor=linkcolor
}
\usepackage[nameinlink,capitalize]{cleveref}
\usepackage{dsfont}
\usepackage{stmaryrd}
\usepackage{amsbsy}
\usepackage{amsthm}

\graphicspath{{images}}
\newtheorem{theorem}{Theorem}
\newtheorem{algo}{Algorithm}
\newtheorem{lemma}{Lemma}

\newtheorem{proposition}{Proposition}

\Crefname{tabular}{Tab.}{Tabs.}
\Crefname{figure}{Fig.}{Fig.}
\Crefname{algo}{Alg.}{Alg.}
\Crefname{theorem}{Thm.}{Thm.}

\usepackage{bbm}

\def\Tr{\mathrm{Tr}}

\def\RR{{\mathbb R}}

\DeclareMathAlphabet{\mathpzc}{OT1}{pzc}{m}{it}

\def\cC{{\mathcal C}}

\def\cE{{\mathcal E}}
\def\cF{{\mathcal F}}

\def\cL{{\mathcal L}}

\def\cN{{\mathcal N}}
\def\cO{{\mathcal O}}

\def\bG{{\mathbf G}}

\def\bI{{\mathbf I}}

\def\bM{{\mathbf M}}

\def\bN{{\mathbf N}}

\def\bv{{\mathbf v}}

\def\bX{{\mathbf X}}

\def\bY{{\mathbf Y}}

\def\bgamma{\boldsymbol{\gamma}}
\def\bGamma{\boldsymbol{\Gamma}}

\def\bSigma{\boldsymbol{\Sigma}}

\def\b1{\boldsymbol{1}}

\def\b{{\underline{b}}}

\newcommand\diag[1]{\mathrm{Diag}\left[#1\right]}
\def\se{{\sigma_e^2}}

\def\BibTeX{{\rm B\kern-.05em{\sc i\kern-.025em b}\kern-.08em
    T\kern-.1667em\lower.7ex\hbox{E}\kern-.125emX}}

\begin{document}

\title{Revisiting CHAMPAGNE: Sparse Bayesian Learning as Reweighted Sparse Coding\\
    \thanks{This work was supported by the French National Agency for Research through the BMWs project (ANR-20-CE45-0018)}
}

\author{\IEEEauthorblockN{Dylan Sechet, Matthieu Kowalski}
    \IEEEauthorblockA{\textit{Laboratoire Interdisciplinaire des Sciences du Numériques} \\
        \textit{Inria, Université Paris-Saclay, CNRS}\\
        Gif-sur-Yvette, France \\
        \{dylan.sechet,matthieu.kowalski\}@inria.fr}
    \and
    \IEEEauthorblockN{Samy Mokhtari, Bruno Torrésani}
    \IEEEauthorblockA{\textit{Aix Marseille Univ, CNRS, I2M}\\
        \textit{Marseille, France} \\
        \{samy.mokhtari, bruno.torresani\}@univ-amu.fr}
}

\maketitle

\begin{abstract}
    Objectif: This paper revisits the CHAMPAGNE algorithm within the Sparse Bayesian Learning (SBL) framework and establishes its connection to reweighted sparse coding. We demonstrate that the SBL objective can be reformulated as a reweighted $\ell_{21}$-minimization problem, providing a more straightforward interpretation of the sparsity mechanism and enabling the design of an efficient iterative algorithm. Additionally, we analyze the behavior of this reformulation in the low signal-to-noise ratio (SNR) regime, showing that it simplifies to a weighted $\ell_{21}$-regularized least squares problem. Numerical experiments validate the proposed approach, highlighting its improved computational efficiency and ability to produce exact sparse solutions, particularly in simulated MEG source localization tasks.
\end{abstract}


\section{Introduction}
Inverse problems focus on reconstructing unknown parameters or sources from incomplete and noisy observed data. In the context of magnetoencephalography (MEG), this involves estimating cortical activity from sensors located at several points around the subject's head~\cite{baillet2017magnetoencephalography}. Such an inverse problem is highly under-determined, where the number of potential sources exceeds the available measurements. Addressing this under-determination requires introducing some regularization techniques or prior information.

Among these strategies, the minimum norm estimate (MNE)~\cite{hamalainen1994interpreting} is a classical approach that favors smooth solutions by minimizing the $\ell_2$ norm of the sources. However, MNE fails to capture the localized nature of brain activity. Sparsity-promoting methods have been developed to address this limitation, leveraging the assumption that brain activity is spatially sparse.

In particular, lasso-type methods~\cite{tibshirani1996regression,chen1994basis}, which include $\ell_1$-regularized least squares and extensions such as the group lasso~\cite{yuan2006model}, enforce sparsity directly on the solution, and have been widely adopted for MEG source localization~\cite{ou2009distributed,gramfort2012mixed}. They use convex penalties, and their optimization problems can be efficiently solved with iterative algorithms, such as the (Fast) Iterative Shrinkage-Thresholding Algorithm ((F)ISTA)~\cite{beck2009fast, tseng2010approximation}. However, lasso-based methods rely heavily on manually tuning regularization parameters, which can be challenging in practice.

In contrast, Bayesian frameworks such as Sparse Bayesian Learning (SBL)~\cite{wipf2004sparse} require minimal tuning, as the only parameter to adjust is the noise variance, which can also be estimated within the framework. SBL introduces a hierarchical Bayesian model where source variances act as hyperparameters driving sparsity, and solutions are estimated iteratively. Among SBL algorithms, \textit{CHAMPAGNE}~\cite{owen2008estimating,hashemi2021unification} has been particularly successful, offering robust solutions by jointly estimating the sources and their variances. However, while CHAMPAGNE induces sparsity in source variances, the corresponding source estimates remain only approximately sparse, with small but not exactly zero components. This introduces the need for manual thresholding.

\paragraph{Contributions and outline}
In this work, we revisit the CHAMPAGNE algorithm within the SBL framework by reformulating it as a reweighted sparse
coding problem. \cref{sec:SBL} introduces the general SBL model and reviews the CHAMPAGNE algorithm.
\cref{sec:main_results} demonstrates that the SBL objective can be expressed as a reweighted sparse coding problem,
which leads to a more efficient algorithm. We also provide a convergence analysis. Furthermore, we examine the behavior
of this reformulation in the low signal-to-noise ratio (SNR) regime. Finally, our numerical experiments validate our
approach, showing improved computational efficiency and achieving exact sparsity in the solutions.

\paragraph{Notations}
Vectors are in bold lowercase (e.g., $\bv$), and matrices in bold uppercase (e.g., $\bM$).
The element at the $i$-th row and $j$-th column of a matrix $\bM$ is denoted by $M[i,j]$,
while $\bM[i,:]$ represents the $i$-th row, and $\bM[:,j]$ denotes the $j$-th column.
Similarly, the $i$-th coordinate of a vector $\bv$ is denoted by $v[i]$ or $v_i$

\section{SBL and CHAMPAGNE algorithm}\label{sec:SBL}
This section presents the general SBL model. For the sake of simplicity, we assume that the noise variance is known. We then focus on the CHAMPAGNE algorithm, which adopts a majorize-minimize strategy relying on a convex upper bound of the SBL objective.

\subsection{SBL model}
In the context of MEG, the objective is to estimate the cortical activity $\bX$ from the sensor measurements $\bY$. This relationship is modeled as follows:
\begin{equation}
    \bY = \bG \bX + \bN,
\end{equation}
where $\bY \in \RR^{M \times T}$ represents the observed sensor data, $\bG \in \RR^{M \times N}$ is the leadfield matrix describing the mapping from cortical sources to sensors, $\bX \in \RR^{N \times T}$ denotes the unknown source activities, and $\bN \in \RR^{M \times T}$ is additive white (in space and time) Gaussian noise with $N[n,t]  \sim \cN(0, \se), \forall n, t.$.

To promote sparsity, SBL models the source activity $\bX$ as a zero-mean multivariate Gaussian with a diagonal covariance matrix $\bGamma$:
\begin{align}
    \bX[:,t] & \sim \cN(\boldsymbol{0}, \bGamma), \quad \bGamma = \diag{\bgamma}, \quad \forall n, \gamma_n \geq 0,
\end{align}
where the hyperparameters $\gamma_n$ encode the variance of each source, indirectly controlling its sparsity. While various priors can be placed on $\gamma_n$, we focus on an exponential prior:
\begin{align}
    \gamma_n \sim \cE(\rho),
\end{align}
a particular case of the classical Gamma prior often used for scale parameters.

To estimate $\bgamma$, SBL proposes to maximize the marginal likelihood of the data (type-II likelihood), which leads to minimizing the following cost function:
\begin{align}
    \cL^{II}(\bgamma) & = -\log p(\bY, \bgamma)  = -\log p(\bY | \bgamma) - \log p(\bgamma)                                                             \\
                      & = \frac{1}{2} \bY^T \bSigma_y(\bgamma)^{-1} \bY + \frac{1}{2} \log |\bSigma_y(\bgamma)| + \rho \sum_{n=1}^N \gamma_n, \nonumber
\end{align}
where:
\begin{equation}
    \label{eq:gamma_y}
    \bSigma_y(\bgamma) = \se \bI_M + \bG \bGamma \bG^T.
\end{equation}

To optimize $\cL^{II}(\bgamma)$, the Expectation-Maximization (EM) algorithm~\cite{wipf2004sparse} alternates updates of the sources $\bX$ and the hyperparameters $\bgamma$. However, it tends to converge slowly. In contrast, the CHAMPAGNE algorithm~\cite{owen2008estimating} utilizes a convex bounding strategy, resulting in significantly faster computations. Hashemi et al.~\cite{hashemi2021unification} provides a unified perspective on these and related SBL algorithms, highlighting CHAMPAGNE as an effective and scalable solution, which is further elaborated in the next section.

\subsection{CHAMPAGNE algorithm}
The CHAMPAGNE algorithm comes from reformulating the original marginal likelihood objective as a minimization problem involving the sources $\bX$, the hyperparameters $\bgamma$, and an auxiliary dual variable $\bv$. The resulting cost function is expressed as:
\begin{align}
    \cC(\bX,\bgamma,\bv) = & \frac{1}{T\se} \|\bY - \bG\bX\|^2 +  \frac{1}{T}
    \sum_{n=1}^N\sum_{t=1}^{T}  \frac{X[n,t]^2}{\gamma_n} \nonumber                  \\
                           & + \bv^T\bgamma - w^*(\bv) + \rho \sum_{n=1}^N \gamma_n,
    \label{eq:crit_champ}
\end{align}
where $w^*(\bv)$ is the concave conjugate of $\log |\bSigma_y(\bgamma)|$:
\begin{equation}\label{eq:wstar}
    w^*(\bv) = \max_{\bgamma}  \bv^T\bgamma -  \log |\bSigma_y(\bgamma)|,
\end{equation}
% and its dual property ensures the equivalence:
% \begin{equation}
%     \frac{1}{2} \log |\bSigma_y(\bgamma)|  = \max_{\bv} \bv^T\bgamma - w^*(\bv).
% \end{equation}
The CHAMPAGNE algorithm alternates between updates of $\bv$, $\bX$, and $\bgamma$, minimizing~\cref{eq:crit_champ} with respect to each variable while keeping the others fixed. The updates are computed iteratively as follows:
\begin{align}
    \bv^{(i)}      & = \diag{\bG^T\bSigma_y^{-1}(\bgamma^{(i-1)})\bG}, \label{eq:v_update}          \\
    \bX^{(i)}      & = \bGamma^{(i-1)}\bG^T\bSigma_y^{-1}(\bgamma^{(i-1)}) \bY, \label{eq:x_update} \\
    \gamma_n^{(i)} & = \frac{\|\bX^{(i)}[n,:]\|}{\sqrt{T (v_n + \rho)}}. \label{eq:gamma_update}
\end{align}
The iterative updates have been shown to converge to a local minimum~\cite{hashemi2021unification}.

\section{SBL by (re)weighted sparse coding}\label{sec:main_results}
This section shows that the original SBL optimization problem can be reformulated as a reweighted sparse coding problem, where sparsity is enforced iteratively through adaptive weights. Then, we analyze the behavior of this reformulation in the low SNR regime, where the log-determinant term simplifies, leading to a weighted $\ell_{21}$-regularized least squares problem.

\subsection{Reformulation and Algorithm}
The following proposition shows that for a fixed $\bv$, the minimization of~\cref{eq:crit_champ} reduces to a weighted (group)-lasso sparse coding problem.

\begin{proposition}
    \label{prop:sparse_coding}
    Minimizing the joint cost function concerning both $\bX$ and $\bgamma$ is equivalent to solving a sparse coding problem for $\bX$. Specifically, we have
    \begin{align*}
        \min_{\bgamma \geq 0, \bX }                                                                               & \frac{ \|\bY - \bG\bX\|^2 }{T\se}
        + \sum_{n=1}^N\sum_{t=1}^T\frac{X[n,t]^2}{T\gamma_n}  + \bv^T\bgamma  + \rho \sum_{n=1}^N \gamma_n        &                                   & \\
                                                                                                                  & =
        \min_{\bX } \frac{1}{T\se} \|\bY - \bG\bX\|^2 +    2\sum_{n=1}^N \sqrt{\frac{\rho + v_n}{T}} \|\bX[n,:]\| &                                   &
    \end{align*}
\end{proposition}

\begin{proof}
    The equivalence is obtained by substituting the optimal solution for $\bgamma$ derived in~\cref{eq:gamma_update} into the original cost function~\cref{eq:crit_champ}.
\end{proof}

Based on this reformulation, the iterative algorithm proceeds as follows
\begin{algo}\label{algo}
    Initialization: $\bgamma^{(0)}\in\RR_+^N$.
    \begin{itemize}
        \item Update $\bv$:
              \begin{equation}
                  \bv^{(i)} = \diag{\bG^T \bSigma_y^{-1}(\bgamma^{(i)}) \bG}.
              \end{equation}
        \item Minimize w.r.t $\bX$
              \begin{equation}
                  \frac{1}{2\se} \|\bY - \bG\bX\|^2 +  \sqrt{T}\sum_{n=1}^N \sqrt{\rho + v_n}\|\bX[n,:]\|\ ,
              \end{equation}
              using $K$ iterations of ISTA. Denotes the output by $\bX^{(i)}$.
        \item Update $\bgamma$:
              \begin{equation}
                  \gamma_n^{(i+1)} = \frac{\|\bX^{(i)}[n,:]\|}{\sqrt{ T(v_n^{(i)} + \rho)}}.
              \end{equation}
    \end{itemize}
\end{algo}

To analyze the convergence properties of~\cref{algo}, we introduce the following auxiliary objective function $\cF$:
\begin{equation}
    \cF(\bX, \bv) = \frac{\|\bY - \bG\bX\|^2 }{T\se} +  2\sum_{n=1}^N \sqrt{\frac{\rho + v_n}{T}} \|\bX[n,:]\| - w^*(\bv).
\end{equation}
with $w^*(\bv)$ given by~\cref{eq:wstar}.
\(\cF\) is directly linked to the original SBL cost function \(\cL^{II}\), as stated in the following lemma.

\begin{lemma}
    \label{lem:stationary}
    If $(\bX, \bv)$ is a stationary point of $\cF$, then $\bgamma$ defined as $\gamma_n = \frac{\|\bX[n,:]\|}{\sqrt{ (\rho + v_n)}}$ is a stationary point of $\cL^{II}$.
\end{lemma}

\begin{proof}
    The equivalence follows directly from the relation:
    \[
        \cF(\bX^{(i)}, \bv^{(i)}) = \cC(\bX^{(i)}, \bgamma^{(i+1)}, \bv^{(i)})\ \forall i\geq 0 .
    \]
    Which is also true for a stationary point $(\bX, \bv)$ of $\cF$ with $\gamma_n = \frac{\|\bX[n,:]\|}{\sqrt{T (\rho + v_n)}}$ for all $n$. Hence, $(\bX, \bgamma, \bv)$ is a stationary point of $\cC$, and then $\bgamma$ satisfies the stationarity conditions of $\cL^{II}$ as shown in~\cite{hashemi2021unification}.
\end{proof}

Next, we show that \(\cF\) decreases monotonically during the iterations.

\begin{lemma}
    \label{lem:F_decrease}
    For all $i\geq 0$, we have
    \[
        \cF(\bX^{(i+1)}, \bv^{(i+1)}) \leq \cF(\bX^{(i)}, \bv^{(i)}) - \frac{L}{2} \|\bX^{(i+1)} - \bX^{(i)}\|^2.
    \]
\end{lemma}

\begin{proof}
    Using the properties of ISTA, we have (see for example~\cite{tseng2010approximation}):
    \begin{equation}
        \cF(\bX^{(i+1)}, \bv^{(i+1)}) \leq \cF(\bX^{(i)}, \bv^{(i+1)}) - \frac{L}{2} \|\bX^{(i+1)} - \bX^{(i)}\|^2.
    \end{equation}
    Furthermore, since $\bv^{(i+1)}$ minimizes $\cF(\bX^{(i+1)}, \bv)$, the result follows.
\end{proof}

By combining the results above, we can state the following theorem on the convergence of the proposed algorithm.

\begin{theorem}
    Any accumulation point of the sequence $\{\bgamma^{(i)}\}$ is a stationary point of $\cL^{II}$.
\end{theorem}

\begin{proof}
    By Lemma~\ref{lem:F_decrease}, the sequence $\{\cF(\bX^{(i)}, \bv^{(i)})\}$ is strictly decreasing. Moreover, since
    $\Sigma\ge\sigma_e^2 I$, the sequence $\{\bv^{(i)}\}$ is necessarily bounded. Hence, $\cF$ is lower bounded, ensuring the convergence to a local minimum $F^*$.
    Finally, the continuity of \(\cF\) ensures that any accumulation point satisfies the stationarity conditions of
    \(\cL^{II}\), as established in Lemma~\ref{lem:stationary}.
\end{proof}

In practice, monitoring the convergence of \(\cL^{II}\) can be numerically challenging due to the division by \(\gamma_n\), which tends to zero for sparse solutions. As  at each iteration \(w^*(\bv^{(i)})\) is given by:
\begin{equation}
    w^*(\bv^{(i)}) = \bv^{(i)T}\bgamma^{(i-1)} -  \log \left| \bSigma_y(\bgamma^{(i-1)}) \right|,
\end{equation}
\(\cF\) can be stably monitored without divisions by \(\gamma_n\).

To accelerate convergence, ISTA can be replaced by FISTA, which improves the convergence rate and reduces the overall computational cost.

\subsection{Low SNR regime}
The low SNR regime is particularly relevant in MEG applications, where the signal-to-noise ratio is often low due to measurement noise and the small amplitude of cortical activity. In this regime, as proposed in~\cite{hashemi2021unification}, the contribution of the data covariance to the log-determinant term in $\cL^{II}$ becomes negligible. Specifically, we have:
\begin{equation}
    \log |\bSigma_y| = \Tr(\bG\bGamma\bG^T) + \cO(SNR).
\end{equation}
Substituting this approximation into the SBL objective, one can show that the resulting cost function $\cL^{low}$ can be expressed as:
\begin{align}
    \cL^{low}(\bX,\bgamma)                                            &
    = \frac{1}{T\se} \|\bY - \bG\bX\|^2
    + \frac{1}{T}  \sum_{n=1}^N\sum_{t=1}^T \frac{X[n,t]^2}{\gamma_n} &                                                & \nonumber \\
                                                                      & + \sum_{n=1}^N (\|\bG[:,n]\|^2+\rho) \gamma_n.
\end{align}
%The term $\|\bG_n\|^2$ represents the squared norm of the $n$-th column of the leadfield matrix $\bG$, capturing the strength of the forward model contribution for each source. 
Minimizing $\cL^{low}$ with respect to both $\bX$ and $\bgamma$ yields the equivalent minimisation problem in $\bX$:
\begin{align}
    \label{eq:min_lowsnr}
    \min_{\bX} \frac{1}{2} \|\bY - \bG\bX\|^2 +  \se\sqrt{T}\sum_{n=1}^N \sqrt{\|\bG[:,n]\|^2 + \rho}\  \|X[n,:]\|\ ,
\end{align}
which shows that in the low SNR regime, the problem reduces to a weighted $\ell_{21}$-regularized least squares problem. The weights depend explicitly, and only, on the norm of the columns of $\bG$ and the regularization parameter $\rho$.

\section{Numerical Results}
In this section, we numerically compare the standard CHAMPAGNE algorithm against \cref{algo}.
We also compare two different implementations of the $\ell_{21}$-solver in \cref{algo}, a standard FISTA solver and celer~\cite{pmlr-v80-massias18a}, a highly optimized solver leveraging a duality-based approach. It should be noted that celer does not provide a solver for reweighted $\ell_{21}$, which causes our setup to introduce a variable change for that implementation. This may impact speed and result quality for celer when $T > 1$.

The experiments are run on an Ubuntu system equipped with an i7-10750H processor and 16 GB of RAM. All benchmarks are managed using benchopt~\cite{benchopt}. The source code to reproduce our results can be found at \href{https://gitlab.inria.fr/dsechet/sbl\_as\_sparse\_coding}{https://gitlab.inria.fr/dsechet/sbl\_as\_sparse\_coding}.

\subsection{Compressed sensing}
We first compare the algorithms on a simple compressed sensing scenario, with $\bG$ a $300 \times 1000$ column-normalized standard Gaussian random matrix, $90 \%$ sparsity, and $T=1$.
Signals are generated with low noise (average SNR of $25$~dB in sensor space). All three implementations show their best performance for $\rho=0$. Performance is aggregated across $1000$ random samples: for each experiment, both $\bG$ and $\bX$ are redrawn.
Performance over time is shown in \cref{fig:csc_t1}. Our algorithm is significantly faster than CHAMPAGNE and converges to a slightly better solution in terms of SNR.

\begin{figure}[h]
    \centering
    \includegraphics[width=\linewidth]{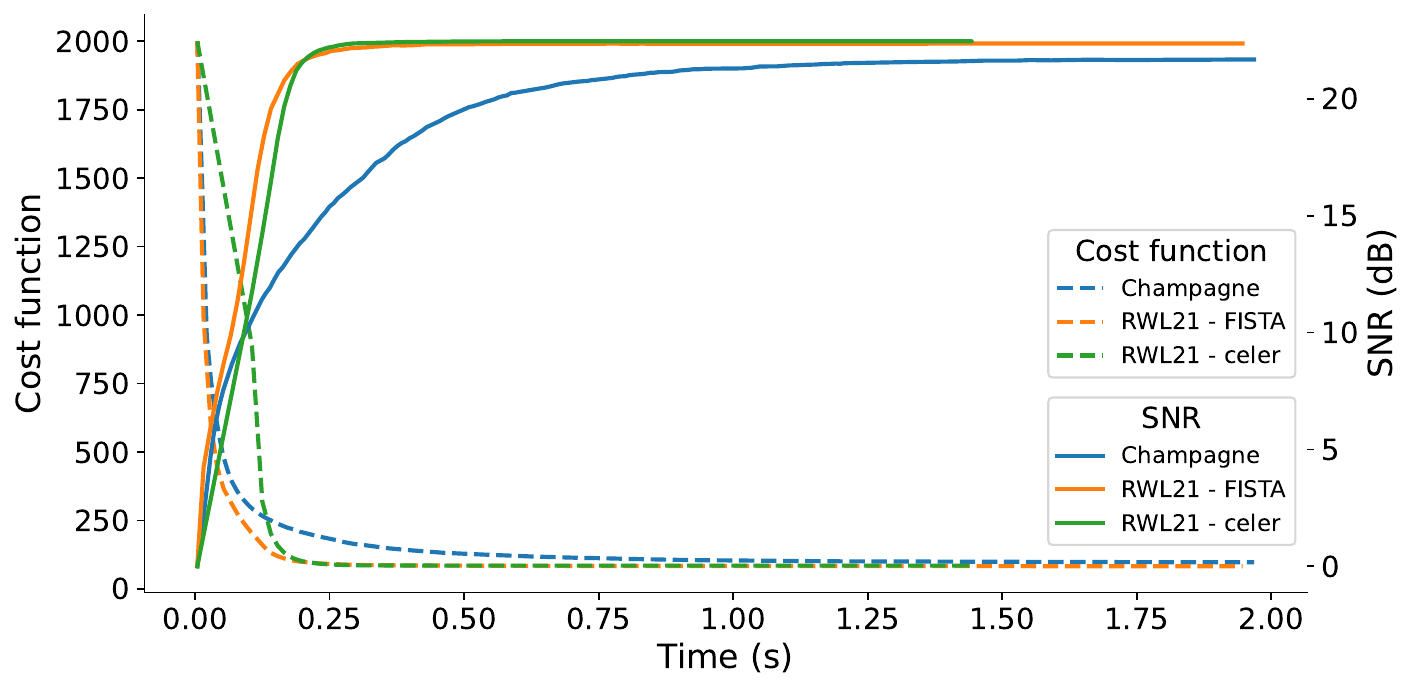}
    \caption{Median SNR and $\cF$  evolution across 1000 compressed sensing problems for $T=1$ and $\rho = 0$}
    \label{fig:csc_t1}
\end{figure}

\subsection{MEG setting}
To evaluate the robustness of our algorithm, we tested it in a realistic MEG inverse problem scenario. MEG inverse
problems usually involve highly sparse solutions in space but are particularly challenging due to their high noise
levels and ill-conditioned sensing matrices.
We use a column-normalized matrix $\bG$ of size $305 \times 1884$, derived from the sample subject in the MNE-python
dataset \cite{GramfortEtAl2013a},
and generate a sparse signal in space, with $0.5 \%$ of active sources. The
active sources are populated for $T=10$ timesteps using a randomized frequency, amplitude, and phase sinusoid. The
noise variance is chosen so the data's SNR is close to $0.3$~dB.
% to emulate the noisy conditions of MEG measures. 

Given the high noise level, we compare the standard SBL algorithm to its low-SNR approximation, where we treat
$\sigma_e$ as a hyperparameter $\sigma_0$, effectively eliminating the need for $\rho$, which we set to zero. In this
regime, the approximate algorithm tends to induce excessive sparsity, often leading to trivial solutions. While
adjusting $\rho$ could mitigate this, our experiments indicate that the optimal weight would typically be negative.  Performance is aggregated across 30 trials, with $\bG$ fixed while $\bX$ and the noise are resampled.

In \cref{fig:snr_03}, we plot the time to convergence of various algorithms, along with the SNR at convergence. Convergence is monitored using the cost function introduced in \cref{eq:min_lowsnr}.
The results highlight a tradeoff between speed and solution quality. While the standard SBL formulations achieve much better solution quality, they converge more slowly than the low-SNR approximation methods.

Celer achieves the fastest convergence in low-SNR mode and also achieves slightly better results than CHAMPAGNE. That tendency is inverted for the standard SBL formulation, with CHAMPAGNE slightly outperforming celer while being faster.
FISTA performs poorly on this problem, converging much slower than the other two implementations, both in SBL and low-SNR mode.
Our choice of using $\sigma_0$ as a hyperparameter instead of directly setting $\sigma_0 = \sigma_e$ in the low-SNR approximation is vindicated, as we can improve the solutions' quality without any significant impact on runtime with an optimal choice of $\sigma_0$.
In \cref{fig:snr_03_evol}, we can see how the SNR evolves over time for both the standard SBL algorithm and the low-SNR
approximation with an optimal $\sigma_0$.

\begin{figure}[h]
    \centering
    \includegraphics[width=\linewidth]{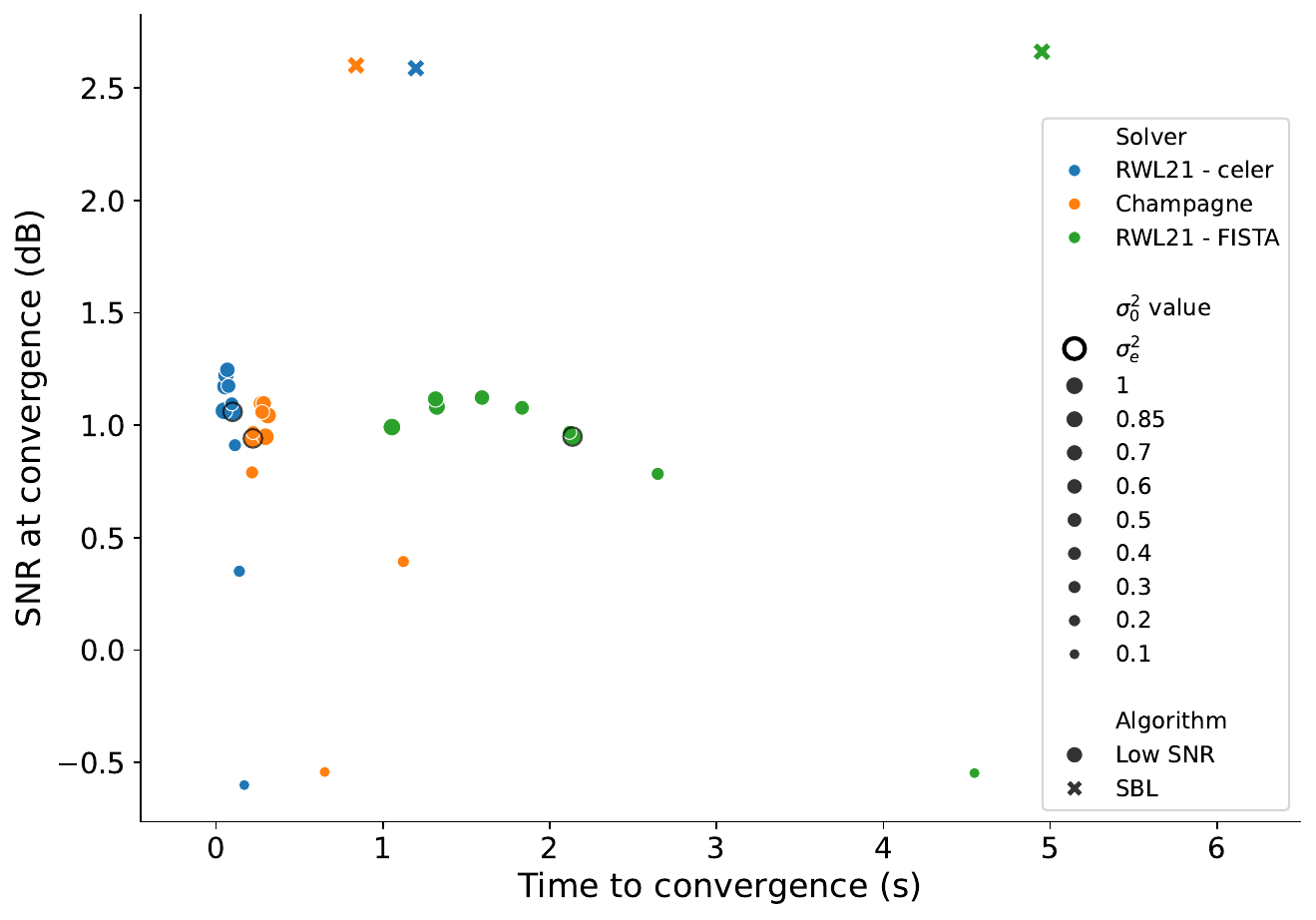}
    \caption{Median SNR and time to convergence accross $30$ MEG problems}
    \label{fig:snr_03}
\end{figure}

\begin{figure}[h]
    \centering
    \includegraphics[width=\linewidth]{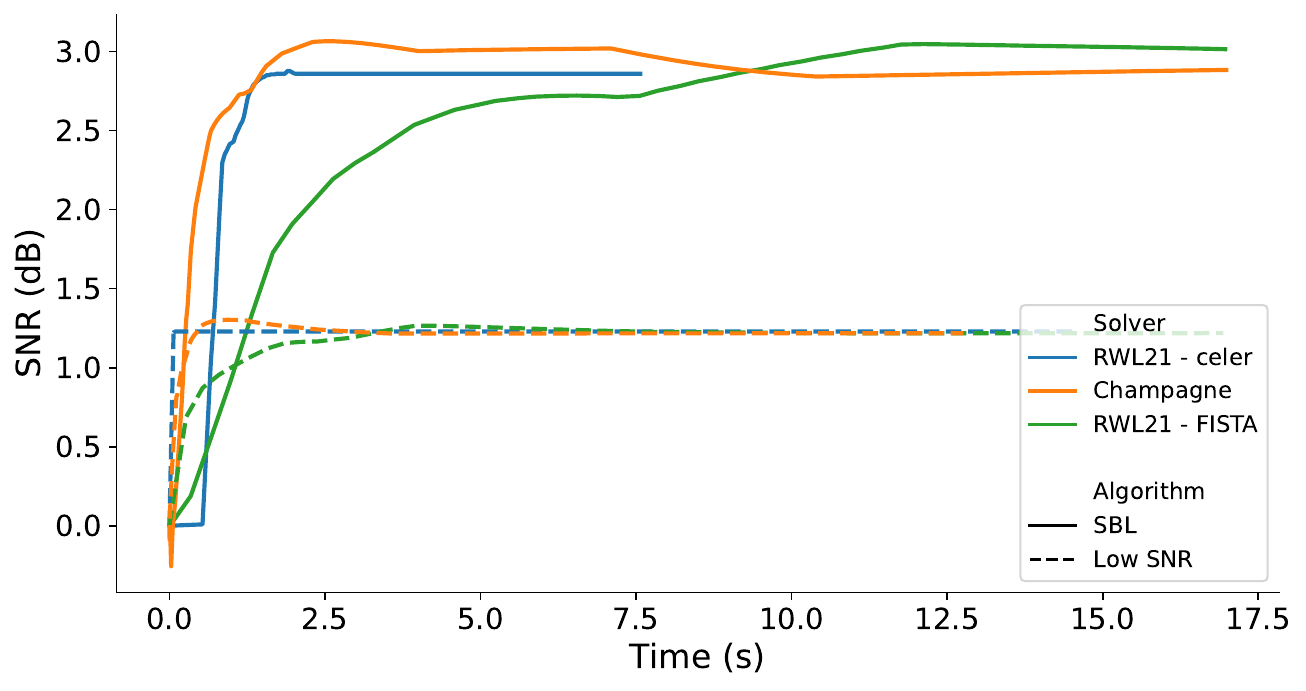}
    \caption{Median SNR evolution across 30 MEG problems. For the low-SNR regime, the best $\sigma_0$ is chosen. It happens to be $\sigma_0^2 = 0.6$ for all solvers.}
    \label{fig:snr_03_evol}
\end{figure}

\section{Conclusion}
We introduced a reweighted sparse coding formulation of the SBL optimization problem, leading to an efficient iterative
algorithm. Our results show that for $T=1$, the proposed approach significantly outperforms CHAMPAGNE regarding
computational efficiency while ensuring exact sparsity. For $T>1$, performance depends on the choice of the iterative
thresholding solver: while FISTA proves suboptimal, using Celer yields competitive results. In the
low-SNR regime, the non-reweighted formulation consistently achieves superior efficiency.
An unexpected observation is that setting $\rho=0$ systematically leads to the best results, which warrants further
investigation. Future work will also focus on analyzing the convergence properties of the iterates by leveraging the
Kurdyka-\L{}ojasiewicz framework.

\newpage

\bibliographystyle{IEEEtran}
% \bibliography{biblio}

\end{document}